\documentclass{article}

\usepackage{amssymb}
\usepackage{eucal}
\usepackage{amsmath}
\usepackage{amsthm}

\numberwithin{equation}{section}
\newtheorem{thm}{\bf Theorem}[section]

\newtheorem{defn}{\bf Definition}[section]

\theoremstyle{remark}
\newtheorem{rem}{\bf Remark}[section]

\title{Stability problem for the torque-free gyrostat by using algebraic methods}
\author{Dan Com\u anescu\\
{\small Department of Mathematics, West University of Timi\c soara}\\
{\small Bd. V. P\^ arvan, No 4, 300223 Timi\c soara, Rom\^ ania}\\
{\small E-mail address: comanescu@math.uvt.ro}}
\date{}

\begin{document}

\maketitle

\begin{abstract}
We apply an algebraic method for studying the stability with respect to a set of conserved quantities for the problem of torque-free gyrostat. If the conditions of this algebraic method are not fulfilled then the Lyapunov stability cannot be decided using the specified set of conserved quantities.
\end{abstract}

\textbf{MSC 2010:} 34D20, 37B25, 70E50, 70H14

\textbf{Keywords:} stability, rigid body, gyrostat

\section{Introduction}

A very important problem in the theory of the differential equations is the problem of the stability. A very useful tool for determining stability of an equilibrium point is Lyapunov's direct method connected with the Lyapunov functions. A natural candidate to be a Lyapunov function is a conserved quantity. In a lot of examples coming from mathematical physics, we identify a set $\{F_1,...,F_k\}$ of conserved quantities. In many situations they are not positive definite functions in the equilibrium points of interest. In this situation, a first step to decide if the equilibrium point is stable is to search a Lyapunov function of the form $\Phi(F_1,...,F_k)$, where $\Phi:\mathbb{R}^k\rightarrow \mathbb{R}$ is a smooth function. A function $\Phi(F_1,...,F_k)$ is a Lyapunov function if and only if it is a positive definite function. In Stability Theory are known some important methods to construct positive definite functions using conserved quantities. We remind the so-called "Chetaev's method" presented in \cite{rouche} and some methods which appeared in the context of Hamilton-Poisson systems. In 1965 have been introduced, see \cite{arnold} the Arnold's method. At the beginning of eighties was developed the Energy-Casimir method (see \cite{holm-marsden-ratiu-weinstein-84}, \cite{holm-marsden-ratiu-weinstein-85}) and in 1998 the paper \cite{ortega-ratiu} present the Ortega-Ratiu method. In \cite{birtea-puta} is proved the equivalence of the Arnold's method, the Energy-Casimir method and the Ortega-Ratiu method.

If there exists, for an equilibrium point, a positive definite function of type $\Phi(F_1,...,F_k)$ we say that the equilibrium point is stable with respect to the set of conserved quantities $\{F_1,...,F_k\}$. In 1958 G.K. Pozharitsky, see \cite{pozharitsky}, had proved that it is sufficient to study the function $||(F_1,...,F_k)||$ in order to decide if an equilibrium point is stable with respect to the set of conserved quantities $\{F_1,...,F_k\}$, (see \cite{rouche}).
Another method to decide if an equilibrium point is stable with respect to the set of conserved quantities $\{F_1,...,F_k\}$, see \cite{rouche}, is given by an algebraic method which reduces to study if the equilibrium point $x_e$ is isolated in the set of all the solutions of the algebraic system $F_1(x)=F_1(x_e),...,F_k(x)=F_k(x_e)$. We also show that if the equilibrium point $x_e$ is not isolated in the set of all the solutions of the algebraic system given above, then it is impossible to construct a Lyapunov function in $x_e$ using the set of conserved quantities $\{F_1,...,F_k\}$.

We apply this algebraic method to decide the stability of an equilibrium point with respect to a set of conserved quantities for the problem of torque-free gyrostat.
In Section 2 are presented some notions and results on Stability Theory which we apply in the study of our example.

In Section 3 we present the mathematical model of a torque-free gyrostat and we give a set of two functionally independent conserved quantities. We find the set of uniform rotations and we first study their stability  with respect to a single conserved quantity. In the cases when the vector of gyrostatic moment is situated along a principal axis of inertia of the gyrostat, we study the stability of an uniform rotation with respect to the set of the two conserved quantities. We prove that an uniform rotation is stable with respect to the given set of conserved quantities if and only if it is stable in the sense of Lyapunov. It is interesting to see that exist some singular cases for which we cannot decide the Lyapunov stability of an uniform rotation using the algebraic method or the linearization method.

In a future study we will apply the algebraic method which is used in this paper for the problem of the rotational motion of a
gyrostat in the presence of an axisymmetric force field. We take advantage of the set of conserved quantities found in \cite{birtea-casu-comanescu}.

\section{Lyapunov's direct method solving algebraic equations}

We consider an open set $D\subset\mathbb{R}^n$ and the locally Lipschitz function $f:D\rightarrow \mathbb{R}^n$ which generates the differential equation
\begin{equation}\label{ecuatie-diferentiala}
    \dot{x}=f(x).
\end{equation}
We denote by $x(\cdot, x_0)$ the maximal solution of the above differential equation which verify the initial condition $x(0,x_0)=x_0$.
A point $x_e\in D$ is an equilibrium point of \eqref{ecuatie-diferentiala} if and only if $f(x_e)=0$. An equilibrium point $x_e\in D$ is stable (or stable in the sense of Lyapunov) if for all $\varepsilon>0$ there exists $\delta>0$ such that for all $y$ in the ball $B(x_e,\delta)$ and $t\geq 0$ we have $||x(t,y)-x_e||<\varepsilon $ (see \cite{perko}).
The most important result for proving stability of an equilibrium point is given by Lyapunov's direct method.

\begin{thm}
 Suppose there exists a continuous function $V:D\rightarrow \mathbb{R}$ satisfying the conditions:
\begin{itemize}
\item [i)] $V(x_e)=0$;
\item [ii)] $V(x)>0$ for $x$ in a neighborhood of $x_e$ and $x\neq x_e$;
\item [iii)] $ t\rightarrow V(x(t,y))$ is a decreasing function for all $y\in D$.
\end{itemize}
Then the equilibrium point $x_e$ is stable.
\end{thm}

A continuous function which satisfies the conditions $i)$ and $ii)$ is called a positive definite function in the equilibrium point $x_e$.
A continuous function $V$ satisfying the hypotheses of the above theorem is called Lyapunov function in the equilibrium point $x_e$. We introduce the following notion of stability.

\begin{defn}
The equilibrium point $x_e$ of \eqref{ecuatie-diferentiala} is stable with respect to the set of conserved quantities $\{F_1,...,F_k\}$ if there exists a continuous function $\Phi:\mathbb{R}^k\rightarrow \mathbb{R}$ such that $x\rightarrow \Phi(F_1,....,F_k)(x)-\Phi(F_1,....,F_k)(x_e)$ is a positive definite function in $x_e$.
\end{defn}

In the conditions of the above definition, the function $x\rightarrow \Phi(F_1,....,F_k)(x)-\Phi(F_1,....,F_k)(x_e)$ is a Lyapunov function in the equilibrium point $x_e$. We have the obvious consequences.
\begin{thm}\label{subset-stability}
Let $x_e$ be an equilibrium point and $\{F_1,...,F_k\}$ be a set of conserved quantities for \eqref{ecuatie-diferentiala}.
\begin{itemize}
\item [(i)] If $x_e$ is stable with respect to the set $\{F_1,...,F_k\}$ then it is stable in the sense of Lyapunov.

\item [(ii)] Let $q\in \{1,...,k\}$ be an integer number. If $x_e$ is stable with respect to $\{F_1,...,F_q\}$, then it is stable with respect to $\{F_1,...,F_k\}$.
    \end{itemize}
\end{thm}

We have the following equivalent conditions for the stability of an equilibrium point with respect to a set of conserved quantities.

\begin{thm}\label{stability}
Let $x_e$ be an equilibrium point of \eqref{ecuatie-diferentiala} and $\{F_1,...,F_k\}$ a set of conserved quantities. The following statements are equivalent:
\begin{itemize}
\item[(i)] $x_e$ is stable with respect to the set of conserved quantities $\{F_1,...,F_k\}$;
\item [(ii)] $x\rightarrow ||(F_1,...,F_k)(x)-(F_1,...,F_k)(x_e)||$ is a positive definite function in $x_e$;
\item [(iii)] the system $F_1(x)=F_1(x_e),...,F_k(x)=F_k(x_e)$ has no root besides $x_e$ in some neighborhood of $x_e$.
\end{itemize}
\end{thm}

In 1958, G.K. Pozharitsky had proved the equivalence between $(i)$ and $(ii)$, see \cite{pozharitsky}, \cite{rouche} pp. 130. Equivalence between $(ii)$ and $(iii)$ appears in \cite{rouche} pp. 151.
In the paper \cite{aeyels}, Aeyels had presented an interesting proof for the implication "$(iii)\Rightarrow x_e$ {\it is Lyapunov stable}".
\medskip

The Theorem \ref{stability} $(iii)$ gives an algebraic method for establishing Lyapunov stability of an equilibrium point. Moreover, it also shows that if the equilibrium point $x_e$ is not isolated in the set of solutions for the algebraic system of equations then it is impossible to construct a Lyapunov function in $x_e$ using the set of conserved quantities $\{F_1,...,F_k\}$. We will apply this algebraic method to study the stability of uniform rotations for a torque-free gyrostat.

\medskip

Using the implicit function theorem we have the following necessary but not sufficient condition for positive definiteness of the function given in Theorem \ref{stability} $(ii)$, see \cite{rouche} pp.151.

\begin{thm}\label{rank-jacobian}
Let $x_e$ be an equilibrium point of \eqref{ecuatie-diferentiala} and $\{F_1,...,F_k\}$ be a set of $\mathcal{C}^1$ conserved quantities. A necessary condition for the stability of $x_e$ with respect to set of conserved quantities $\{F_1,...,F_k\}$ is that the jacobian matrix $\frac{\partial(F_1,...,F_k)}{\partial x}(x_e)$ be of rank strictly smaller then $k$.
\end{thm}
\medskip

For the case of one conserved quantity, i.e. $k=1$, we have the well known result.

\begin{thm}\label{stability-one-conserved}
Let $x_e$ be an equilibrium point of \eqref{ecuatie-diferentiala} and $F$ a conserved quantity.
The following statements are equivalent:
\begin{itemize}
\item[(i)] $x_e$ is stable with respect to the conserved quantity $F$;
\item[(ii)] $x_e$ is a strict local extremum of $F$.
\end{itemize}
\end{thm}

\section{The stability of the uniform rotations of a torque-free gyrostat}

For the problem of torque-free gyrostat we find the set of uniform rotations and we study their stability  with respect to a conserved quantity. In the cases when the vector of the gyrostatic moment is situated along a principal axis of inertia of the gyrostat, we study the stability of an uniform rotation with respect to the set formed by two conserved quantities. Except two singular cases the Lyapunov stability problem for the free-torque gyrostat can be decided using the algebraic method with two conserved quantities and the linearization method. In the singular cases we decide the Lyapunov stability by studying the dynamics in an invariant set.

The equation for the rotation of a torque-free gyrostat is given by (see \cite{birtea-casu-comanescu},\cite{wittenburg})
\begin{equation}\label{torque-free-omega}
\mathbb{I}\dot{\vec{\omega}}=(\mathbb{I} \vec{\omega}+\vec{\mu})\times \vec{\omega},
\end{equation}
where $\vec{\omega}$ is the angular velocity, and $\mathbb{I}$ is the inertia tensor and $\vec{\mu}$ is the constant vector of gyrostatic moment. We denote by $I_1,I_2$ and $I_3$ the principal moments of inertia and suppose that $I_1>I_2>I_3$. If we use the angular momentum vector $\vec{M}=\mathbb{I} \vec{\omega}$ then the equation becomes
\begin{equation}\label{torque-free-rotations}
    \dot{\vec{M}}=(\vec{M}+\vec{\mu})\times \mathbb{I}^{-1}\vec{M}.
\end{equation}
It is easy to see that for the above dynamic we have two conserved quantities
$$F_1=\frac{1}{2}\vec{M}\cdot \mathbb{I}^{-1}\vec{M},\,\,\,F_2=\frac{1}{2}(\vec{M}+\vec{\mu})\cdot (\vec{M}+\vec{\mu}).$$
Next, we find the set of the uniform rotations.
In the paper \cite{puta-comanescu} was considered the differential equation
\begin{equation}\label{puta-comanescu}
    \dot{\vec{N}}=\vec{N}\times \mathbb{I}^{-1}\vec{N}+\vec{a}\times \vec{N},
\end{equation}
where $\vec{a}\in\mathbb{R}^3$. This equation is equivalent with the torque-free gyrostat equation \eqref{torque-free-rotations} where $\vec{a}=-\mathbb{I}^{-1}\vec{\mu}$ and making the change of variable $\vec{M}=\vec{N}-\vec{\mu}$.
According to \cite{puta-comanescu} the equilibrium points of \eqref{puta-comanescu} are of the following types:
\begin{itemize}
  \item [i.] $\vec{N}_1=(0,0,0)$;
  \item [ii.] $\vec{N}_{2}=(\frac{\mu_1}{1-\lambda I_1}, \frac{\mu_2}{1-\lambda I_2}, \frac{\mu_3}{1-\lambda I_3})$ for $\lambda\in \mathbb{R}\backslash \{\frac{1}{I_1},\frac{1}{I_2},\frac{1}{I_3}\}$;
  \item [iii.] $\vec{N}_{3}=(\alpha, \frac{\mu_2I_1}{I_1-I_2}, \frac{\mu_3I_1}{I_1-I_3})$ if $\mu_1=0$ and $\alpha\in \mathbb{R}$;
  \item [iv.] $\vec{N}_{4}=(\frac{\mu_1I_2}{I_2-I_1}, \alpha, \frac{\mu_3I_2}{I_2-I_3})$ if $\mu_2=0$ and $\alpha\in \mathbb{R}$;
  \item [v.] $\vec{N}_{5}=(\frac{\mu_1I_3}{I_3-I_1}, \frac{\mu_2I_3}{I_3-I_2}, \alpha)$ if $\mu_3=0$ and $\alpha\in \mathbb{R}$.
\end{itemize}
Consequently, the uniform rotations of the torque-free equation \eqref{torque-free-rotations} are of the types:
\begin{itemize}
  \item [i.] $\vec{M}_1=(-\mu_1,-\mu_2,-\mu_3)$;
  \item [ii.] $\vec{M}_{2}=(\frac{\lambda I_1}{1-\lambda I_1}\mu_1, \frac{\lambda I_2}{1-\lambda I_2}\mu_2, \frac{\lambda I_3}{1-\lambda I_3}\mu_3)$ for $\lambda\in \mathbb{R}\backslash \{\frac{1}{I_1},\frac{1}{I_2},\frac{1}{I_3}\}$;
  \item [iii.] $\vec{M}_{3}=(\beta, \frac{I_2}{I_1-I_2}\mu_2, \frac{I_3}{I_1-I_3}\mu_3)$ if $\mu_1=0$ and $\beta\in \mathbb{R}$;
  \item [iv.] $\vec{M}_{4}=(\frac{I_1}{I_2-I_1}\mu_1, \beta, \frac{I_3}{I_2-I_3}\mu_3)$ if $\mu_2=0$ and $\beta\in \mathbb{R}$;
  \item [v.] $\vec{M}_{5}=(\frac{I_1}{I_3-I_1}\mu_1, \frac{I_2}{I_3-I_2}\mu_2, \beta)$ if $\mu_3=0$ and $\beta\in \mathbb{R}$.
\end{itemize}
Analogous considerations are made in \cite{wittenburg}, pp. 78-80, for finding the uniform rotations of the system \eqref{torque-free-omega}.
\medskip

First we study the stability of an uniform rotation with respect to one conserved quantity.

\begin{thm}\label{stability-one-conserved-quantity}
For the uniform rotations of a torque-free gyrostat we have:
\begin{itemize}
\item [(i)] The unique uniform rotation which is stable with respect to $F_1$ is $(0,0,0)$. This uniform rotation is of type $\vec{M}_2$ obtained for $\lambda=0$.
\item [(ii)] The unique uniform rotation which is stable with respect to $F_2$ is $\vec{M}_1=(-\mu_1,-\mu_2,-\mu_3)$.
\end{itemize}
\end{thm}

\begin{proof}
$(i)$ The uniform rotation $(0,0,0)$ is the unique strict local extremum of the conserved quantity $F_1$. Using the Theorem \ref{stability-one-conserved} we obtain the result.

\noindent $(ii)$ The uniform rotation $(-\mu_1,-\mu_2,-\mu_3)$ is the unique strict local extremum of the conserved quantity $F_2$. Using the Theorem \ref{stability-one-conserved} we obtain the enounced result.
\end{proof}
The uniform rotations found in the above theorem are the only uniform rotations which Lyapunov stability can be proved by using only one of the conserved quantities. For the rest of the uniform rotations it is necessary to consider both conserved quantities.
Next, we study the stability of the uniform rotations with respect to the set conserved quantities $\{F_1,F_2\}$.
By direct calculus we obtain that the set of uniform rotations coincide with the set where the jacobian matrix $\frac{\partial(F_1,F_2)}{\partial \vec{M}}$ has the rank strictly smaller $2$ and consequently, the necessary condition of Theorem \ref{rank-jacobian} is fulfilled.

In what follows we restrict ourselves to the cases for which the vector of gyrostatic moment $\vec{\mu}$ is situated along a principal axis of inertia of the gyrostat.

\subsection{The case $\mu_2=\mu_3=0$}
In this case we have the following types of uniform rotations:
\begin{itemize}
  \item [i.] $\vec{M}_{1-2}=(q,0,0)$ where $q\in \mathbb{R}$;
  \item [ii.] $\vec{M}_{4}=(\frac{I_1}{I_2-I_1}\mu_1, q,0)$ where $q\in \mathbb{R}^*$;
  \item [iii.] $\vec{M}_{5}=(\frac{I_1}{I_3-I_1}\mu_1, 0, q)$ where $q\in \mathbb{R}^*$.
\end{itemize}
First, we study the solutions of the algebraic system
$$F_1(\vec{M})=F_1(\vec{M}_e),\,\,\,F_2(\vec{M})=F_2(\vec{M}_e),$$
where $\vec{M}_e$ is an uniform rotation. The above system of algebraic equations has the form:
\begin{equation}\label{first-axis}
    \left\{
      \begin{array}{ll}
        \frac{M_2^2}{I_2}+\frac{M_3^2}{I_3}=\frac{M_{1e}^2}{I_1}+\frac{M_{2e}^2}{I_2}+\frac{M_{3e}^2}{I_3}-\frac{M_{1}^2}{I_1} \\
        M_2^2+M_3^2=M_{1e}^2+M_{2e}^2+M_{3e}^2-M_1^2-2\mu_1M_1
      \end{array}
    \right.
\end{equation}
 where the unknowns are $M_1,M_2$ and $M_3$. The system has at least the solution $(M_{1e},M_{2e},M_{3e})$. We want to see if this solution is isolated in the set of all the solutions of the algebraic system.
For our study is preferable to change the variable $M_1$ with $x=M_1-M_{1e}$. The algebraic system \eqref{first-axis} becomes
\begin{equation}\label{first-axis-changed}
    \left\{
      \begin{array}{ll}
        \frac{M_2^2}{I_2}+\frac{M_3^2}{I_3}=\frac{M_{2e}^2}{I_2}+\frac{M_{3e}^2}{I_3}-\frac{x^2+2xM_{1e}}{I_1} \\
        M_2^2+M_3^2=M_{2e}^2+M_{3e}^2-(x^2+2xM_{1e})-2\mu_1(x+M_{1e})
      \end{array}
    \right.
\end{equation}
with the unknowns $x,M_2$ and $M_3$. The system has at least the solution $(0,M_{2e},M_{3e})$.
The solution $(M_{1e},M_{2e},M_{3e})$ of \eqref{first-axis} is isolated in the set of all the solutions of this system if and only if the solution $(0,M_{2e},M_{3e})$ of \eqref{first-axis-changed} is isolated in the set of corresponding solutions.
If we use the unknowns $M_2^2$ and $M_3^2$, then we have a linear system. Using Cramer's rule we can find the solutions of \eqref{first-axis-changed}.
\medskip

\noindent {\bf I. The uniform rotation of type $\vec{M}_{1-2}$.}
The solutions of \eqref{first-axis-changed} verifies
\begin{equation}\label{first-axis-rezolvat}
\left\{
    \begin{array}{ll}
      M_2^2=-\frac{2xI_2}{I_1(I_2-I_3)}\left(\frac{1}{2}x(I_1-I_3)+q(I_1-I_3)+I_1\mu_1\right) \\
      M_3^2=\frac{2xI_3}{I_1(I_2-I_3)}\left (\frac{1}{2}x(I_1-I_2)+q(I_1-I_2)+I_1\mu_1\right )
    \end{array}
  \right.
\end{equation}
\medskip

\noindent {\bf I.1.} If $q=-\frac{I_1\mu_1}{I_1-I_3}$, then
the system \eqref{first-axis-rezolvat} becomes
\begin{equation}\label{first-axis-rezolvat-1}
\left\{
    \begin{array}{ll}
      M_2^2=-\frac{x^2I_2(I_1-I_3)}{I_1(I_2-I_3)} \\
      M_3^2=\frac{2xI_3}{I_1(I_2-I_3)}\left (\frac{1}{2}x(I_1-I_2)+\frac{I_1(I_2-I_3)}{I_1-I_3}\mu_1\right )
    \end{array}
  \right.
\end{equation}
By our hypotheses we have $I_1>I_2>I_3$ and if $(x,M_2,M_3)$ is a solution of \eqref{first-axis-rezolvat-1}, then $M_2^2\leq 0$. We deduce that $(0,0,0)$ is the unique solution of the above system and consequently, it is isolated in the set of all the solutions.
\medskip

\noindent {\bf I.2.} If $q=-\frac{I_1\mu_1}{I_1-I_2}$, then
the system \eqref{first-axis-rezolvat} becomes
\begin{equation}\label{first-axis-rezolvat-2}
\left\{
    \begin{array}{ll}
      M_2^2=-\frac{2xI_2}{I_1(I_2-I_3)}\left (\frac{1}{2}x(I_1-I_3)+\frac{I_1(I_2-I_3)}{I_1-I_2}\mu_1\right ) \\
      M_3^2=\frac{x^2I_3(I_1-I_2)}{I_1(I_2-I_3)}
    \end{array}
  \right.
\end{equation}
For $|x|$ sufficiently small we have
$$\text{sgn} \left(-\frac{2xI_2}{I_1(I_2-I_3)}\left (\frac{1}{2}x(I_1-I_3)+\frac{I_1(I_2-I_3)}{I_1-I_2}\mu_1\right )  \right )=
-\text{sgn}(x)\cdot\text{sgn}(\mu_1).$$
For every $|x|$ sufficiently small such that $\text{sgn}(x)=-\text{sgn}{\mu_1}$ we obtain a solution of \eqref{first-axis-rezolvat-2} and consequently, we have that $(0,0,0)$ is not an isolated solution in the set of all the solutions.
\medskip

\noindent {\bf I.3.} The case when $q\neq -\frac{I_1\mu_1}{I_1-I_2}$ and $q\neq -\frac{I_1\mu_1}{I_1-I_3}$.
For $|x|$ sufficiently small the terms in the righthand side of the system \eqref{first-axis-rezolvat} have the properties
 $$\text{sgn}\left(-\frac{2xI_2}{I_1(I_2-I_3)}\left(\frac{1}{2}x(I_1-I_3)+q(I_1-I_3)+I_1\mu_1\right)\right)=-\text{sgn}(x)\cdot\text{sgn}(q(I_1-I_3)+I_1\mu_1),$$
$$\text{sgn}\left( \frac{2xI_3}{I_1(I_2-I_3)}\left (\frac{1}{2}x(I_1-I_2)+q(I_1-I_2)+I_1\mu_1\right )\right )=\text{sgn}(x)\cdot\text{sgn}(q(I_1-I_2)+I_1\mu_1).$$
If $\text{sgn}(q(I_1-I_3)+I_1\mu_1)\cdot\text{sgn}(q(I_1-I_2)+I_1\mu_1)>0$, then exists $r>0$ such that a solution of the form $(x,M_2,M_3)$ which verify $x\neq 0$ has the property $|x|>r$. In this case, the solution $(0,0,0)$ of the system \eqref{first-axis-rezolvat} is an isolated solution in the set of all the solutions.

If $\text{sgn}(q(I_1-I_3)+I_1\mu_1)\cdot\text{sgn}(q(I_1-I_2)+I_1\mu_1)<0$, then for every $|x|$ sufficiently small we have that the solutions of the system \eqref{first-axis-rezolvat} are of the form $(x,M_2,M_3)$. We obtain that the solution $(0,0,0)$ of the system \eqref{first-axis-rezolvat} is not isolated in the set of all the solutions.
\medskip

\noindent {\bf II. The uniform rotation of type $\vec{M}_4$.}
In this case the system \eqref{first-axis-changed} is equivalent with the following system
\begin{equation}\label{first-axis-rezolvat-M4}
\left\{
    \begin{array}{ll}
      M_2^2=q^2-\frac{I_2(I_1-I_3)}{I_1(I_2-I_3)}x^2+\frac{2I_2\mu_1}{I_1-I_2}x \\
      M_3^2=\frac{x^2I_3(I_1-I_2)}{I_1(I_2-I_3)}
    \end{array}
  \right.
\end{equation}
As before, for every $|x|$ sufficiently small we have a solution of the above system which is of the form $(x,M_2,M_3)$ and consequently, $(0,q,0)$ is not isolated in the set of all the solutions of \eqref{first-axis-rezolvat-M4}.
\medskip

\noindent {\bf III. The uniform rotation of type $\vec{M}_5$.}
The system \eqref{first-axis-changed} is equivalent with the following system
\begin{equation}\label{first-axis-rezolvat-M5}
\left\{
    \begin{array}{ll}
M_2^2=-\frac{x^2I_2(I_1-I_3)}{I_1(I_2-I_3)} \\
      M_3^2=q^2+\frac{I_3(I_1-I_2)}{I_1(I_2-I_3)}x^2+\frac{2I_3\mu_1}{I_1-I_3}x
    \end{array}
  \right.
\end{equation}
The solutions of the above system are $(0,0,q)$ and $(0,0,-q)$. The solution $(0,0,q)$ is isolated in the set of all the solutions of \eqref{first-axis-rezolvat-M5}.
\medskip

Summarizing, we obtain the following result.

\begin{thm} For a vector of gyrostatic moment along the first axis of inertia we have the following stability results.
\begin{itemize}
\item [(i)] An uniform rotation of type $\vec{M}_{1-2}=(q,0,0)$ is stable with respect to the set of conserved quantities $\{F_1,F_2\}$ if and only if $q\in (-\infty, -\frac{I_1\mu_1}{I_1-I_2})\cup [-\frac{I_1\mu_1}{I_1-I_3},\infty)$ and $\mu_1>0$ or $q\in (-\infty, -\frac{I_1\mu_1}{I_1-I_3})\cup [-\frac{I_1\mu_1}{I_1-I_2},\infty)$ and $\mu_1<0$.
\item [(ii)] An uniform rotation of type $\vec{M}_4=(\frac{I_1}{I_2-I_1}\mu_1, q,0)$ with $q\neq 0$ is not stable with respect to the set of conserved quantities $\{F_1,F_2\}$.
\item [(iii)] An uniform rotation of type $\vec{M}_5=(\frac{I_1}{I_3-I_1}\mu_1, 0, q)$ with $q\neq 0$ is stable with respect to the set of conserved quantities $\{F_1,F_2\}$.

\end{itemize}
\end{thm}

We prove that the uniform rotations which are not stable with respect to the set $\{F_1,F_2\}$ are Lyapunov unstable.
In the paper \cite{puta-comanescu} is proved that the equilibrium points of the system \eqref{puta-comanescu} have the properties:
\begin{itemize}
\item [(i)] for $\lambda\in (\frac{1}{I_2}, \frac{1}{I_3})$ an equilibrium point of type $\vec{N}_2$ is spectrally unstable;
\item [(ii)] an equilibrium point of type $\vec{N}_4$ is spectrally unstable when $\alpha\neq 0$.
\end{itemize}
Consequently, we have:
\begin{itemize}
\item [(i)] for $q\in (-\frac{I_1\mu_1}{I_1-I_2}, -\frac{I_1\mu_1}{I_1-I_3})$ and $\mu_1>0$ or for $q\in (-\frac{I_1\mu_1}{I_1-I_3}, -\frac{I_1\mu_1}{I_1-I_2})$ and $\mu_1<0$ an uniform rotation of type $\vec{M}_{1-2}$ is spectrally unstable and consequently, it is unstable in the sense of Lyapunov;
\item [(ii)] an uniform rotation of type $\vec{M}_4$ is spectrally unstable and also it is unstable in the sense of Lyapunov.
\end{itemize}

The Lyapunov stability or instability of $\vec{M}_e=(-\frac{I_1\mu_1}{I_1-I_2},0,0)$ cannot be decided using the set of conserved quantities $\{F_1,F_2\}$ or using the linearization method. This uniform rotation is spectrally stable and it is not stable  with respect to the set of of conserved quantities $\{F_1,F_2\}$. The instability in the sense of Lyapunov of this uniform rotation will be proved by studying the dynamics on the invariant set $$\mathcal{M}=\{\vec{M}\,|\,F_1(\vec{M})=F_1(\vec{M_e}),\,\,F_2(\vec{M})=F_2(\vec{M_e})\}.$$

\begin{thm}\label{stabilitate-singular-23}
The uniform rotation $(-\frac{I_1\mu_1}{I_1-I_2},0,0)$ is unstable in the sense of Lyapunov.
\end{thm}

\begin{proof}
The projection of the vectorial differential equation \eqref{torque-free-rotations} on the first axis using the variables $x,M_2$ and $M_3$ is
$$\dot{x}=(\frac{1}{I_3}-\frac{1}{I_2})M_2M_3.$$
By using \eqref{first-axis-rezolvat-2} we have
$$\dot{x}^2=-\frac{2x^3(I_1-I_2)}{I_1^2(I_2-I_3)^2}\left (\frac{1}{2}x(I_1-I_3)+\frac{I_1(I_2-I_3)}{I_1-I_2}\mu_1\right ).$$
First we consider the case $\mu_1>0$. Suppose that we have $x(0)>-\frac{2I_1(I_2-I_3)}{(I_1-I_2)(I_1-I_3)}\mu_1$ and $M_2(0)M_3(0)<0$. Consequently, we obtain that $\dot{x}(0)<0$. In this case there exists $t^*>0$ such that $x(t^*)=-\frac{2I_1(I_2-I_3)}{(I_1-I_2)(I_1-I_3)}\mu_1$ which implies that our uniform rotation is unstable.
In the case $\mu_1<0$ we have analogous considerations.
\end{proof}

\begin{rem}
In this case an uniform rotation is stable with respect to the set of conserved quantities $\{F_1,F_2\}$ if and only if it is stable in the sense of Lyapunov.
\end{rem}
\subsection{The case $\mu_1=\mu_3=0$}

We have the following types of uniform rotations:
\begin{itemize}
  \item [i.] $\vec{M}_{1-2}=(0,q,0)$ where $q\in \mathbb{R}$;
  \item [ii.] $\vec{M}_{3}=(q,\frac{I_2}{I_1-I_2}\mu_2, 0)$ where $q\in \mathbb{R}^*$;
  \item [iii.] $\vec{M}_{5}=(0,\frac{I_2}{I_3-I_2}\mu_2, q)$ where $q\in \mathbb{R}^*$.
\end{itemize}
Using the method of the previous section and by analogous calculations we obtain the following result.
\begin{thm} For a vector of gyrostatic moment along the second axis of inertia we have the following stability results.
\begin{itemize}
\item [(i)] An uniform rotation of type $\vec{M}_{1-2}=(0,q,0)$ is stable with respect to the set of conserved quantities $\{F_1,F_2\}$ if and only if $q\in [-\frac{I_2\mu_2}{I_2-I_3},\frac{I_2\mu_2}{I_1-I_2}]$ and $\mu_2>0$ or $q\in [\frac{I_2\mu_2}{I_1-I_2},-\frac{I_2\mu_2}{I_2-I_3}]$ and $\mu_2<0$.
\item [(ii)] An uniform rotation of type $\vec{M}_3=(q,\frac{I_2}{I_1-I_2}\mu_2, 0)$ with $q\neq 0$ is stable with respect to the set of conserved quantities $\{F_1,F_2\}$.
\item [(iii)] An uniform rotation of type $\vec{M}_5=(0,\frac{I_2}{I_3-I_2}\mu_2, q)$ with $q\neq 0$ is stable with respect to the set of conserved quantities $\{F_1,F_2\}$.

\end{itemize}
\end{thm}
In the paper \cite{puta-comanescu} is proved that an equilibrium point of type $\vec{N}_2$ is spectrally unstable for $\lambda\in (\frac{1}{I_1}, \frac{1}{I_3})$. Consequently, for $q\in \mathbb{R}\backslash [-\frac{I_2\mu_2}{I_2-I_3},\frac{I_2\mu_2}{I_1-I_2}]$ and $\mu_2>0$ or $q\in \mathbb{R}\backslash [\frac{I_2\mu_2}{I_1-I_2},-\frac{I_2\mu_2}{I_2-I_3}]$ and $\mu_2<0$ an uniform rotation of type $\vec{M}_{1-2}=(0,q,0)$ is spectrally unstable and also it is unstable in the sense of Lyapunov.

In this case the stability (in the sense of Lyapunov) can be decided using the stability with respect to the set of conserved quantities $\{F_1,F_2\}$ and the linearization method. An uniform rotation is stable with respect to the set of conserved quantities $\{F_1,F_2\}$ if and only if it is stable in the sense of Lyapunov.

\subsection{The case $\mu_1=\mu_2=0$}

We have the following types of uniform rotations:
\begin{itemize}
  \item [i.] $\vec{M}_{1-2}=(0,0,q)$ where $q\in \mathbb{R}$;
  \item [ii.] $\vec{M}_{3}=(q,0,\frac{I_3}{I_1-I_3}\mu_3)$ where $q\in \mathbb{R}^*$;
  \item [iii.] $\vec{M}_{4}=(0,q,\frac{I_3}{I_2-I_3}\mu_3)$ where $q\in \mathbb{R}^*$.
\end{itemize}

\begin{thm} For a vector of gyrostatic moment along the third axis of inertia we have the following stability results.
\begin{itemize}
\item [(i)] An uniform rotation of type $\vec{M}_{1-2}=(0,0,q)$ is stable with respect to the set of conserved quantities $\{F_1,F_2\}$ if and only if $q\in (-\infty, \frac{I_3\mu_3}{I_1-I_3}]\cup (\frac{I_3\mu_3}{I_2-I_3},\infty)$ and $\mu_3>0$ or $q\in (-\infty, \frac{I_3\mu_3}{I_2-I_3})\cup [\frac{I_3\mu_3}{I_1-I_3},\infty)$ and $\mu_3<0$.
\item [(ii)] An uniform rotation of type $\vec{M}_3=(q,0,\frac{I_3}{I_1-I_3}\mu_3)$ with $q\neq 0$ is stable with respect to the set of conserved quantities $\{F_1,F_2\}$.
\item [(iii)] An uniform rotation of type $\vec{M}_4=(0,q,\frac{I_3}{I_2-I_3}\mu_3)$ with $q\neq 0$ is not stable with respect to the set of conserved quantities $\{F_1,F_2\}$.

\end{itemize}
\end{thm}
In the paper \cite{puta-comanescu} is proved that the equilibrium points of the system \eqref{puta-comanescu} has the properties:
\begin{itemize}
\item [(i)] for $\lambda\in (\frac{1}{I_1}, \frac{1}{I_2})$ an equilibrium point of type $\vec{N}_2$ is spectrally unstable;
\item [(ii)] an equilibrium point of type $\vec{N}_4$ is spectrally unstable when $\alpha\neq 0$.
\end{itemize}
These results implies:
\begin{itemize}
\item [(i)] for $q\in (\frac{I_3\mu_3}{I_1-I_3}, \frac{I_3\mu_3}{I_2-I_3})$ and $\mu_3>0$ or for $q\in (\frac{I_3\mu_3}{I_2-I_3}, \frac{I_3\mu_3}{I_1-I_3})$ and $\mu_3<0$ an uniform rotation of type $\vec{M}_{1-2}$ is spectrally unstable and consequently, it is unstable in the sense of Lyapunov;
\item [(ii)] an uniform rotation of type $\vec{M}_4$ is spectrally unstable and also it is unstable in the sense of Lyapunov.
\end{itemize}

The Lyapunov stability or instability of $\vec{M}_e=(0,0,\frac{I_3\mu_3}{I_2-I_3})$ cannot be decided using the set of conserved quantities $\{F_1,F_2\}$ or using the linearization method. This uniform rotation is spectrally stable and it is not stable  with respect to the set of of conserved quantities $\{F_1,F_2\}$. The instability in the sense of Lyapunov of this uniform rotation will be proved by studying the dynamics on the invariant set $$\mathcal{M}=\{\vec{M}\,|\,F_1(\vec{M})=F_1(\vec{M_e}),\,\,F_2(\vec{M})=F_2(\vec{M_e})\}.$$
The proof is analogous to the proof of the Theorem \ref{stabilitate-singular-23}.
\begin{thm}
The uniform rotation $(0,0,\frac{I_3\mu_3}{I_2-I_3})$ is unstable in the sense of Lyapunov.
\end{thm}
\medskip

An uniform rotation is stable with respect to the set of conserved quantities $\{F_1,F_2\}$ if and only if it is stable in the sense of Lyapunov.
\medskip

\noindent {\bf Acknowledgement} The author would like to acknowledge the many conversations with Petre Birtea.

\end{document}